\documentclass[conference,a4paper]{IEEEtran}

\usepackage[english]{babel}
\usepackage[utf8]{inputenc}
\usepackage[T1]{fontenc}
\usepackage{hyphenat}
\usepackage{dsfont}

\usepackage{times}
\usepackage{amsmath,amssymb,amstext,amsfonts,mathrsfs,amsthm,mathtools,cases}
\usepackage{comment}

\usepackage{graphicx}                           
\graphicspath{{.}{./Figures/}}

\usepackage[nolist,printonlyused]{acronym}      

\usepackage[sort]{cite}                         

\usepackage{enumerate}                          

\usepackage{multirow}
\usepackage{tabularx}
\usepackage{booktabs}
\newcolumntype{C}{>{\centering\arraybackslash}X}
\newcolumntype{R}{>{\flushright\arraybackslash}X}
\newcolumntype{L}{>{\flushleft\arraybackslash}X}

\usepackage{algorithmic}
\usepackage[ruled]{algorithm}
\algsetup{indent=2ex}
\usepackage{eqparbox}

\usepackage{setspace}
\usepackage{xspace} 
\newcommand*{\unit}[2]{\mbox{\ensuremath{#1\,\mathrm{#2}}}}
\newcommand*{\scale}[2][4]{\scalebox{#1}{$#2$}}

\usepackage{float}
\floatname{algorithm}{\footnotesize Algorithm}

\usepackage{url}

\usepackage[normalem]{ulem}

\newcommand*{\figref}[1]{Figure~\ref{#1}}

\newcommand*{\tabref}[1]{Table~\ref{#1}}

\theoremstyle{definition}
\newtheorem{thm}{Theorem}

\newtheorem{lem}{Lemma}

\newtheorem{defi}{Definition}

\usepackage{color}

\newcommand{\comm}[1]{}

\DeclareMathAlphabet{\mathppl}{T1}{ppl}{m}{it}
\DeclareMathAlphabet{\mathphv}{T1}{phv}{m}{it}
\DeclareMathAlphabet{\mathpzc}{T1}{pzc}{m}{it}
\newlength{\norlen} 

\newcommand{\Abs}[1]{\left\vert #1 \right\vert}

\newcommand{\Field}[1]{\mathbb{\uppercase{#1}}}

\newcommand{\Herm}[1]{{#1}^{\mathrm{H}}}

\newcommand{\Mt}[1]{\mathbf{#1}}
\newcommand{\Norm}[1]{\left\Vert #1 \right\Vert}

\newcommand{\NormTwo}[1]{\Norm{#1}_{2}}

\newcommand{\Set}[1]{\mathcal{\uppercase{#1}}}

\newcommand{\Transp}[1]{{#1}^{\mathrm{T}}}

\newcommand{\Vt}[1]{\mathbf{\lowercase{#1}}}
\newcommand{\mtA}{\Mt{A}}

\newcommand{\mtI}{\Mt{I}}

\newcommand{\mtX}{\Mt{X}}

\newcommand{\vtH}{\Vt{H}}

\newcommand{\vtU}{\Vt{U}}
\newcommand{\vtV}{\Vt{V}}
\newcommand{\vtW}{\Vt{W}}
\newcommand{\vtX}{\Vt{X}}
\newcommand{\vtY}{\Vt{Y}}
\newcommand{\vtZ}{\Vt{Z}}

\newcommand{\vtEta}{\Vt{\boldsymbol{\eta}}}

\newcommand{\vtZero}{\Vt{0}}

\newcommand{\fdC}{\Field{C}}

\newcommand{\fdE}{\Field{E}}

\newcommand{\fdR}{\Field{R}}

\newcommand{\stU}{\Set{U}}

\newcommand{\bbE}{\mathbb{E}}

\newcommand{\calD}{\mathcal{D}}

\newcommand{\calO}{\mathcal{O}}

\newcommand{\calS}{\mathcal{S}}

\IEEEoverridecommandlockouts

\allowdisplaybreaks[4]

\begin{document}
\date{\today}


\title{Simultaneous Wireless Information and Power Transfer for Federated Learning}

\author{
Jos\'{e} Mairton Barros da Silva Jr.$^{\star}$,
\thanks{The simulations were performed on resources provided by the Swedish National Infrastructure for Computing 
(SNIC) at PDC Centre for High Performance Computing (PDC-HPC).
This work has received funding from the European Research Council (ERC) under the European Union’s Horizon 2020 
research and innovation programme (Grant agreement No. 819819).}
Konstantinos Ntougias$^{\ddagger}$, Ioannis Krikidis$^{\ddagger}$, G\'{a}bor Fodor$^{\star\dagger}$, Carlo 
Fischione${^\star}$\\
$^\star$School of Electrical Engineering and Computer Science, KTH Royal Institute of Technology, Stockholm, Sweden\\
$^\ddagger$Department of Electrical and Computer Engineering, University of Cyprus, Nicosia, Cyprus\\
$^\dagger$Ericsson Research, Kista, Sweden\\
Email: $^\star\{$jmbdsj, gaborf, carlofi$\}@$kth.se, $^\ddagger\{$ntougias.konstantinos, krikidis.ioannis$\}@$ucy.ac.cy
}



\begin{acronym}[LTE-Advanced]
  \acro{2G}{Second Generation}
  \acro{3-DAP}{3-Dimensional Assignment Problem}
  \acro{3G}{3$^\text{rd}$~Generation}
  \acro{3GPP}{3$^\text{rd}$~Generation Partnership Project}
  \acro{4G}{4$^\text{th}$~Generation}
  \acro{5G}{5$^\text{th}$~Generation}
  \acro{AA}{Antenna Array}
  \acro{AC}{Admission Control}
  \acro{AD}{Attack-Decay}
  \acro{ADC}{analog-to-digital converter}
  \acro{ADSL}{Asymmetric Digital Subscriber Line}
  \acro{AHW}{Alternate Hop-and-Wait}
  \acro{AMC}{Adaptive Modulation and Coding}
  \acro{AP}{Access Point}
  \acro{APA}{Adaptive Power Allocation}
  \acro{ARMA}{Autoregressive Moving Average}
  \acro{ATES}{Adaptive Throughput-based Efficiency-Satisfaction Trade-Off}
  \acro{AWGN}{additive white Gaussian noise}
  \acro{BB}{Branch and Bound}
  \acro{BCD}{block coordinate descent}
  \acro{BD}{Block Diagonalization}
  \acro{BER}{Bit Error Rate}
  \acro{BF}{Best Fit}
  \acro{BFD}{bidirectional full duplex}
  \acro{BLER}{BLock Error Rate}
  \acro{BPC}{Binary Power Control}
  \acro{BPSK}{Binary Phase-Shift Keying}
  \acro{BRA}{Balanced Random Allocation}
  \acro{BS}{base station}
  \acro{BSUM}{block successive upper-bound minimization}
  \acro{CAP}{Combinatorial Allocation Problem}
  \acro{CAPEX}{Capital Expenditure}
  \acro{CBF}{Coordinated Beamforming}
  \acro{CBR}{Constant Bit Rate}
  \acro{CBS}{Class Based Scheduling}
  \acro{CC}{Congestion Control}
  \acro{CDF}{Cumulative Distribution Function}
  \acro{CDMA}{Code-Division Multiple Access}
  \acro{CL}{Closed Loop}
  \acro{CLPC}{Closed Loop Power Control}
  \acro{CNR}{Channel-to-Noise Ratio}
  \acro{CNN}{convolutional neural network}
  \acro{CPA}{Cellular Protection Algorithm}
  \acro{CPICH}{Common Pilot Channel}
  \acro{CoMP}{Coordinated Multi-Point}
  \acro{CQI}{Channel Quality Indicator}
  \acro{CRM}{Constrained Rate Maximization}
	\acro{CRN}{Cognitive Radio Network}
  \acro{CS}{Coordinated Scheduling}
  \acro{CSI}{channel state information}
  \acro{CSIR}{channel state information at the receiver}
  \acro{CUE}{Cellular User Equipment}
  \acro{D2D}{device-to-device}
  \acro{DAC}{digital-to-analog converter}
  \acro{DC}{direct current}
  \acro{DCA}{Dynamic Channel Allocation}
  \acro{DE}{Differential Evolution}
  \acro{DFT}{Discrete Fourier Transform}
  \acro{DIST}{Distance}
  \acro{DL}{downlink}
  \acro{DMA}{Double Moving Average}
  \acro{DMRS}{demodulation reference signal}
  \acro{D2DM}{D2D Mode}
  \acro{DMS}{D2D Mode Selection}
  \acro{DPC}{Dirty Paper Coding}
  \acro{DRA}{Dynamic Resource Assignment}
  \acro{DSA}{Dynamic Spectrum Access}
  \acro{DSM}{Delay-based Satisfaction Maximization}
  \acro{ECC}{Electronic Communications Committee}
  \acro{EFLC}{Error Feedback Based Load Control}
  \acro{EI}{Efficiency Indicator}
  \acro{EMPOWER}{learning }
  \acro{eNB}{Evolved Node B}
  \acro{EPA}{Equal Power Allocation}
  \acro{EPC}{Evolved Packet Core}
  \acro{EPS}{Evolved Packet System}
  \acro{E-UTRAN}{Evolved Universal Terrestrial Radio Access Network}
  \acro{ES}{Exhaustive Search}
  \acro{FD}{full duplex}
  \acro{FDD}{frequency division duplex}
  \acro{FDM}{Frequency Division Multiplexing}
  \acro{FER}{Frame Erasure Rate}
  \acro{FF}{Fast Fading}
  \acro{FL}{federated learning}
  \acro{FSB}{Fixed Switched Beamforming}
  \acro{FST}{Fixed SNR Target}
  \acro{FTP}{File Transfer Protocol}
  \acro{GA}{Genetic Algorithm}
  \acro{GBR}{Guaranteed Bit Rate}
  \acro{GLR}{Gain to Leakage Ratio}
  \acro{GOS}{Generated Orthogonal Sequence}
  \acro{GPL}{GNU General Public License}
  \acro{GRP}{Grouping}
  \acro{HARQ}{Hybrid Automatic Repeat Request}
  \acro{HD}{half-duplex}
  \acro{HMS}{Harmonic Mode Selection}
  \acro{HOL}{Head Of Line}
  \acro{HSDPA}{High-Speed Downlink Packet Access}
  \acro{HSPA}{High Speed Packet Access}
  \acro{HTTP}{HyperText Transfer Protocol}
  \acro{ICMP}{Internet Control Message Protocol}
  \acro{ICI}{Intercell Interference}
  \acro{ID}{Identification}
  \acro{IETF}{Internet Engineering Task Force}
  \acro{ILP}{Integer Linear Program}
  \acro{JRAPAP}{Joint RB Assignment and Power Allocation Problem}
  \acro{UID}{Unique Identification}
  \acro{IID}{Independent and Identically Distributed}
  \acro{IIR}{Infinite Impulse Response}
  \acro{ILP}{Integer Linear Problem}
  \acro{IMT}{International Mobile Telecommunications}
  \acro{INV}{Inverted Norm-based Grouping}
  \acro{IoT}{Internet of Things}
  \acro{IP}{Integer Programming}
  \acro{IPv6}{Internet Protocol Version 6}
  \acro{ISD}{Inter-Site Distance}
  \acro{ISI}{Inter Symbol Interference}
  \acro{ITU}{International Telecommunication Union}
  \acro{JAFM}{joint assignment and fairness maximization}
  \acro{JAFMA}{joint assignment and fairness maximization algorithm}
  \acro{JOAS}{Joint Opportunistic Assignment and Scheduling}
  \acro{JOS}{Joint Opportunistic Scheduling}
  \acro{JP}{Joint Processing}
	\acro{JS}{Jump-Stay}
  \acro{KKT}{Karush-Kuhn-Tucker}
  \acro{L3}{Layer-3}
  \acro{LAC}{Link Admission Control}
  \acro{LA}{Link Adaptation}
  \acro{LC}{Load Control}
  \acro{LOS}{line of sight}
  \acro{LP}{Linear Programming}
  \acro{LTE}{Long Term Evolution}
	\acro{LTE-A}{\ac{LTE}-Advanced}
  \acro{LTE-Advanced}{Long Term Evolution Advanced}
  \acro{M2M}{Machine-to-Machine}
  \acro{MAC}{medium access control}
  \acro{MANET}{Mobile Ad hoc Network}
  \acro{MC}{Modular Clock}
  \acro{MCS}{Modulation and Coding Scheme}
  \acro{MDB}{Measured Delay Based}
  \acro{MDI}{Minimum D2D Interference}
  \acro{MF}{Matched Filter}
  \acro{MG}{Maximum Gain}
  \acro{MH}{Multi-Hop}
  \acro{MIMO}{Multiple Input Multiple Output}
  \acro{MINLP}{mixed integer nonlinear programming}
  \acro{MIP}{Mixed Integer Programming}
  \acro{MISO}{multiple input single output}
  \acro{MLWDF}{Modified Largest Weighted Delay First}
  \acro{MME}{Mobility Management Entity}
  \acro{MMSE}{minimum mean squared error}
  \acro{MOS}{Mean Opinion Score}
  \acro{MPF}{Multicarrier Proportional Fair}
  \acro{MRA}{Maximum Rate Allocation}
  \acro{MR}{Maximum Rate}
  \acro{MRC}{maximum ratio combining}
  \acro{MRT}{maximum ratio transmission}
  \acro{MRUS}{Maximum Rate with User Satisfaction}
  \acro{MS}{Mode Selection}
  \acro{MSE}{mean squared error}
  \acro{MSI}{Multi-Stream Interference}
  \acro{MTC}{Machine-Type Communication}
  \acro{MTSI}{Multimedia Telephony Services over IMS}
  \acro{MTSM}{Modified Throughput-based Satisfaction Maximization}
  \acro{MU-MIMO}{Multi-User Multiple Input Multiple Output}
  \acro{MU}{Multi-User}
  \acro{NAS}{Non-Access Stratum}
  \acro{NB}{Node B}
	\acro{NCL}{Neighbor Cell List}
  \acro{NLP}{Nonlinear Programming}
  \acro{NLOS}{non-line of sight}
  \acro{NMSE}{Normalized Mean Square Error}
  \acro{NORM}{Normalized Projection-based Grouping}
  \acro{NP}{non-polynomial time}
  \acro{NRT}{Non-Real Time}
  \acro{NSPS}{National Security and Public Safety Services}
  \acro{O2I}{Outdoor to Indoor}
  \acro{OFDMA}{Orthogonal Frequency Division Multiple Access}
  \acro{OFDM}{Orthogonal Frequency Division Multiplexing}
  \acro{OFPC}{Open Loop with Fractional Path Loss Compensation}
	\acro{O2I}{Outdoor-to-Indoor}
  \acro{OL}{Open Loop}
  \acro{OLPC}{Open-Loop Power Control}
  \acro{OL-PC}{Open-Loop Power Control}
  \acro{OPEX}{Operational Expenditure}
  \acro{ORB}{Orthogonal Random Beamforming}
  \acro{JO-PF}{Joint Opportunistic Proportional Fair}
  \acro{OSI}{Open Systems Interconnection}
  \acro{PAIR}{D2D Pair Gain-based Grouping}
  \acro{PAPR}{Peak-to-Average Power Ratio}
  \acro{P2P}{Peer-to-Peer}
  \acro{PC}{Power Control}
  \acro{PCI}{Physical Cell ID}
  \acro{PDCCH}{physical downlink control channel}
  \acro{PDD}{penalty dual decomposition}
  \acro{PDF}{Probability Density Function}
  \acro{PER}{Packet Error Rate}
  \acro{PF}{Proportional Fair}
  \acro{P-GW}{Packet Data Network Gateway}
  \acro{PL}{Pathloss}
  \acro{PRB}{Physical Resource Block}
  \acro{PROJ}{Projection-based Grouping}
  \acro{ProSe}{Proximity Services}
  \acro{PS}{phase shifter}
  \acro{PSO}{Particle Swarm Optimization}
  \acro{PUCCH}{physical uplink control channel}
  \acro{PZF}{Projected Zero-Forcing}
  \acro{QAM}{Quadrature Amplitude Modulation}
  \acro{QoS}{quality of service}
  \acro{QPSK}{Quadri-Phase Shift Keying}
  \acro{RAISES}{Reallocation-based Assignment for Improved Spectral Efficiency and Satisfaction}
  \acro{RAN}{Radio Access Network}
  \acro{RA}{Resource Allocation}
  \acro{RAT}{Radio Access Technology}
  \acro{RATE}{Rate-based}
  \acro{RB}{resource block}
  \acro{RBG}{Resource Block Group}
  \acro{REF}{Reference Grouping}
  \acro{RF}{radio frequency}
  \acro{RLC}{Radio Link Control}
  \acro{RM}{Rate Maximization}
  \acro{RNC}{Radio Network Controller}
  \acro{RND}{Random Grouping}
  \acro{RRA}{Radio Resource Allocation}
  \acro{RRM}{Radio Resource Management}
  \acro{RSCP}{Received Signal Code Power}
  \acro{RSRP}{reference signal receive power}
  \acro{RSRQ}{Reference Signal Receive Quality}
  \acro{RR}{Round Robin}
  \acro{RRC}{Radio Resource Control}
  \acro{RSSI}{received signal strength indicator}
  \acro{RT}{Real Time}
  \acro{RU}{Resource Unit}
  \acro{RUNE}{RUdimentary Network Emulator}
  \acro{RV}{Random Variable}
  \acro{SAC}{Session Admission Control}
  \acro{SCM}{Spatial Channel Model}
  \acro{SC-FDMA}{Single Carrier - Frequency Division Multiple Access}
  \acro{SD}{Soft Dropping}
  \acro{S-D}{Source-Destination}
  \acro{SDPC}{Soft Dropping Power Control}
  \acro{SDMA}{Space-Division Multiple Access}
  \acro{SDR}{semidefinite relaxation}
  \acro{SDP}{semidefinite programming}
  \acro{SER}{Symbol Error Rate}
  \acro{SES}{Simple Exponential Smoothing}
  \acro{S-GW}{Serving Gateway}
  \acro{SGD}{stochastic gradient descent}
  \acro{SINR}{signal-to-interference-plus-noise ratio}
  \acro{SI}{self-interference}
  \acro{SIP}{Session Initiation Protocol}
  \acro{SISO}{Single Input Single Output}
  \acro{SIMO}{Single Input Multiple Output}
  \acro{SIR}{Signal to Interference Ratio}
  \acro{SLNR}{Signal-to-Leakage-plus-Noise Ratio}
  \acro{SMA}{Simple Moving Average}
  \acro{SNR}{signal to noise ratio}
  \acro{SOC}{second order cone}
  \acro{SOCP}{second order cone programming}
  \acro{SORA}{Satisfaction Oriented Resource Allocation}
  \acro{SORA-NRT}{Satisfaction-Oriented Resource Allocation for Non-Real Time Services}
  \acro{SORA-RT}{Satisfaction-Oriented Resource Allocation for Real Time Services}
  \acro{SPF}{Single-Carrier Proportional Fair}
  \acro{SRA}{Sequential Removal Algorithm}
  \acro{SRS}{sounding reference signal}
  \acro{SU-MIMO}{Single-User Multiple Input Multiple Output}
  \acro{SU}{Single-User}
  \acro{SVD}{Singular Value Decomposition}
  \acro{SVM}{support vector machine}
  \acro{SWIPT}{simultaneous wireless information and power transfer}
  \acro{TCP}{Transmission Control Protocol}
  \acro{TDD}{time division duplex}
  \acro{TDMA}{time division multiple access}
  \acro{TNFD}{three node full duplex}
  \acro{TETRA}{Terrestrial Trunked Radio}
  \acro{TP}{Transmit Power}
  \acro{TPC}{Transmit Power Control}
  \acro{TTI}{transmission time interval}
  \acro{TTR}{Time-To-Rendezvous}
  \acro{TSM}{Throughput-based Satisfaction Maximization}
  \acro{TU}{Typical Urban}
  \acro{UE}{user equipment}
  \acro{UEPS}{Urgency and Efficiency-based Packet Scheduling}
  \acro{UL}{uplink}
  \acro{UMTS}{Universal Mobile Telecommunications System}
  \acro{URI}{Uniform Resource Identifier}
  \acro{URM}{Unconstrained Rate Maximization}
  \acro{VR}{Virtual Resource}
  \acro{VoIP}{Voice over IP}
  \acro{WAN}{Wireless Access Network}
  \acro{WCDMA}{Wideband Code Division Multiple Access}
  \acro{WF}{Water-filling}
  \acro{WiMAX}{Worldwide Interoperability for Microwave Access}
  \acro{WINNER}{Wireless World Initiative New Radio}
  \acro{WLAN}{Wireless Local Area Network}
  \acro{WMMSE}{weighted minimum mean square error}
  \acro{WMPF}{Weighted Multicarrier Proportional Fair}
  \acro{WPF}{Weighted Proportional Fair}
  \acro{WSN}{Wireless Sensor Network}
  \acro{WWW}{World Wide Web}
  \acro{XIXO}{(Single or Multiple) Input (Single or Multiple) Output}
  \acro{ZF}{zero-forcing}
  \acro{ZMCSCG}{Zero Mean Circularly Symmetric Complex Gaussian}
\end{acronym}

\maketitle

\begin{abstract}
In the Internet of Things, learning is one of most prominent tasks. 
In this paper, we consider an Internet of Things scenario where federated learning is used with simultaneous 
transmission of model data and wireless power. 
We investigate the trade-off between the number of communication rounds and communication round time while harvesting 
energy to compensate the energy expenditure. 
We formulate and solve an optimization problem by considering the number of local iterations on devices, the time to 
transmit-receive the model updates, and to harvest sufficient energy.
Numerical results indicate that maximum ratio transmission and zero-forcing beamforming for the optimization of the 
local iterations on devices substantially boost the test accuracy of the learning task.
Moreover, maximum ratio transmission instead of zero-forcing provides the best test accuracy and communication round 
time trade-off for various energy harvesting percentages.
Thus, it is possible to learn a model quickly with few communication rounds without depleting the battery.
\end{abstract}


\section{Introduction}\label{sec:intro}

Internet of Things (\acs{IoT}\acused{IoT}) devices will arguably generate large amount of data, which can be of great 
use for prediction and inference in several applications of major societal interest~\cite{Ericsson2020}. 
Thus, it is natural to push machine learning solution into the \ac{IoT} devices~\cite{Zhou2019}. 
However, machine learning is traditionally conceived in centralized settings, where all data is available at one 
location.
When we use machine learning on distributed wireless scenarios, such as the \ac{IoT} applications, we encounter new 
challenges for both learning (such as heterogeneity of the data generated by the distributed devices and 
privacy issues~\cite{McMahan2017}), and communication (such as communication efficiency, interference channel, 
fading, and battery limitations~\cite{Hellstrom2020}).

Federated learning (\acs{FL}\acused{FL}) is a recently proposed distributed machine learning technique that attempts to 
overcome some of these challenges, since it enhances the privacy by sending only model updates instead of raw data, 
and considers the heterogeneity of devices data and the communication efficiency in the learning 
convergence~\cite{McMahan2017}.
For \ac{IoT} scenarios, the limited energy storage capacity and high maintenance cost of devices battery is the main 
challenge~\cite{Clerckx2021}.
To this end, \ac{FL} in \ac{IoT} scenarios must consider the battery of the devices as one of the most 
important aspects, together with the impact of communication efficiency, interference, and fading.

One promising solution to overcome the energy limitations in \ac{IoT} is energy harvesting, which allows devices to 
harvest \acl{RF} energy when communicating with an edge server, e.g., a \acl{BS} or access point~\cite{Clerckx2019}.
Hence, the use of energy harvesting for \ac{IoT} devices with \ac{FL} would be a perfect combination.
However, how to allow the devices to harvest sufficient energy to train a \ac{FL} model is largely an open question.
This idea is highly novel and, to the best of our knowledge, there is only one similar work in the 
literature~\cite{Zeng2021}.
In~\cite{Zeng2021}, the authors consider a wireless power transfer scenario served by power-beacons using \ac{FL}.
Such a work shows that higher density of power beacons improves the learning convergence, and that the mini-batch size 
and processor clock frequency is inversely proportional to the energy spent in local computation.

Different from the authors in~\cite{Zeng2021}, we investigate the trade-off among the number of communication rounds 
between an edge server and the \ac{IoT} devices, and the communication round time necessary to train, transmit, and 
receive the model while harvesting sufficient energy to compensate part of the total energy spent.
This trade-off is important to understand in which conditions we can train a \ac{FL} task without depleting the battery 
of the devices.
We consider a scenario with a multi-antenna edge server using the \ac{SWIPT} technology with \ac{IoT} devices using 
\ac{FL} to simultaneously train a learning model, while communicating with an edge server (see 
Figure~\ref{fig:scenario_federated}).
Moreover, we consider the use of \mbox{FedProx}~\cite{Li2020b}, a recent generalization of \ac{FL} that allows to 
optimize the number of local iterations at each device, while guaranteeing convergence to (non-)convex learning tasks.
We propose an optimization problem to minimize the number of communication rounds and communication round time while 
optimizing the number of local iterations, the time to transmit/receive, and to harvest a 
percentage of the total energy spent at each round and device.
\begin{figure}
\centering
\includegraphics[width=0.75\linewidth,trim=0mm 0mm 0mm 0mm,clip]{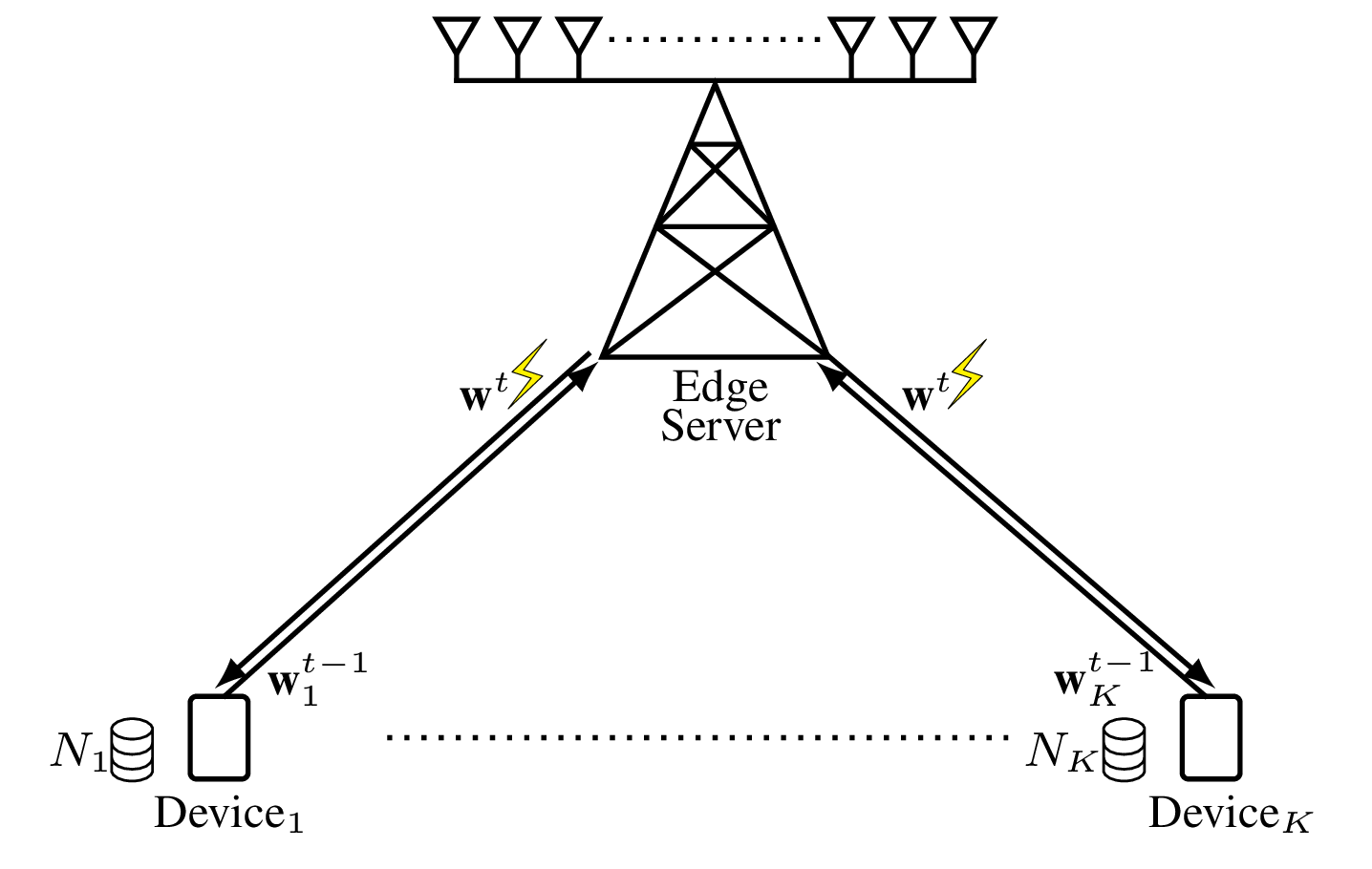}
\caption{An example of a network employing \ac{FL} and \ac{SWIPT} with $K$ devices. The devices send the weights 
$\vtW_k^{t-1}$ in the \acl{UL}, while the edge server sends the aggregated weight $\vtW^t$ and energy in the \acl{DL}.}
\label{fig:scenario_federated}
\end{figure}

We provide the solution to this problem while considering low transmit power at devices and fixed beamforming 
techniques, e.g., \ac{MRT} and \ac{ZF}, at the edge server.
The results indicate that the learning test accuracy using either \ac{MRT} or \ac{ZF} with the optimization of the 
local number of iterations outperform a solution without such optimization.
We also show that \ac{MRT} vastly outperforms \ac{ZF} in terms of minimum communication round time for all the 
percentage of the energy harvesting required. 
Therefore, we show that it is possible to use \ac{FL} in \ac{IoT} scenarios without depleting the battery of the 
devices.

\emph{Notation:}
Vectors and matrices are denoted by bold lower and upper case letters, respectively;
$\Herm{\mtA}$ is the Hermitian of $\mtA$; $\mtI_{K}$ is the identity matrix of dimension $K$; and by $\fdC$ the complex 
field.
We denote expectation by $\bbE[\cdot]$; the uniform distribution with minimum and maximum parameters $a,b$, 
respectively, is denoted by $\stU(a,b)$; the gradient of $f$ at $\vtX$ by $\nabla f(\vtX)\in\fdR^n$; 
and the $2-$norm of a vector $\vtX$ by $\Norm{\vtX}$.
\section{System Model}

We consider a general \ac{IoT} scenario with an edge server equipped with $M$ antennas and $K$ single-antenna 
devices (see Figure~\ref{fig:scenario_federated}). 
The devices represent sensors gathering data for an \ac{FL} task and a power splitting \ac{SWIPT} receiver 
architecture~\cite{Clerckx2019,Timotheou2014}. 
In the following, we define the communication, energy, time, and learning models.
\subsection{Communication Model}\label{sub:comm_model}
Let $\vtH_{k}\in\fdC^{M\times 1}$ denote the \acl{UL} complex channel vector comprising small- and large-scale fading 
that includes the effects of multipath, shadowing, and path-loss between device $k$ and the edge server. 
The transmit power of device $k$ at the \acl{UL} is denoted as $p_k$, which is assumed fixed.
The received signal at the edge server is given by
\begin{equation}
\overline{\vtY}^u = \sum\nolimits_{k=1}^{K} \vtH_k\sqrt{p_k}s_k^u + \vtEta^u \in \mathds{C}^{M \times 1},
\end{equation}
where $s_k^u$ is the transmitted data symbol by device $k$ with zero mean and unit power; and $\vtEta^u\in\fdC^{M\times 
1}$ is the \acl{AWGN} with covariance matrix $\bbE\{ \vtEta^u\Herm{\vtEta^u} \} = \sigma^2\mtI_{M\times M}$.
The \ac{SINR} of the \acl{UL} device $k$ at the edge server is denoted as
\begin{equation}\label{eq:sinr_ul}
\Gamma_k^u = \frac{p_k \Abs{\Herm{\vtU}_k\vtH_k}^2}{\sum_{j\neq k} p_j \Abs{\Herm{\vtU}_k\vtH_j}^2 + \sigma^2},
\end{equation}
where $\vtU_k$ is the \acl{UL} combiner at the edge server and is assumed fixed during the coherence time of the 
channel.
In the \acl{UL}, the achievable rate of device $k$ is given by
\begin{equation}
R_k^u = B_c\log_2\left(1 + \Gamma_k^u\right),
\end{equation}
where $B_c$ denotes the channel bandwidth of each device.

In the \acl{DL}, the edge server aggregates the weight updates from all $k$ devices and unicasts the aggregated updates 
to each device $k$.
The edge server applies a unicast transmission instead of multicast to increase the beamforming gain when harvesting 
the energy.
The transmitted beamforming vector for device $k$ is denoted as $\vtV_k\in\fdC^{M\times 1}$, which is assumed fixed 
according to \ac{MRT} or \ac{ZF}.
The received signal at device $k$ is given by
\begin{equation}\label{eq:rx_symbol_DL}
\overline{y}_k^d = \Herm{\vtH}_k\vtV_k s_k^d + \sum\nolimits_{j\neq k} \Herm{\vtH}_k\vtV_j s_j^d + \eta^d,
\end{equation}
where $s_k^d$ is the transmitted data symbol to device $k$ with zero mean and unit power; and $\eta^d\in\fdC^{1}$ is 
the \acl{AWGN} with variance $\bbE\{ \Abs{\eta^d}^2 \} = \sigma^2$.
Note that we allow for different $s_k^d$ as transmitted symbols in Eq.~\eqref{eq:rx_symbol_DL}, which 
reflect different coding between the devices.
We denote by $P_k^r$ the received power at device $k$, which is given by
\begin{equation}\label{eq:rec_power}
P_k^r = \sum\nolimits_{j=1}^{K} \Abs{\Herm{\vtH}_k\vtV_j}^2 + \sigma^2.
\end{equation}
Note that depending on the beamforming $\vtV_k$, the received power changes drastically due to the cancellation or not 
of the multi-user interference in \ac{ZF} and \ac{MRT}, respectively.
At device $k$, the receiver splits its received signal into the data detection and energy storage circuits.
Let us denote the power splitting parameter for device $k$ as $\delta_k\in (0,1)$, which is fixed and indicates that 
$\delta_k P_k^r$ is directed to the data decoding unit and $(1 - \delta_k) P_k^r$ to the energy harvesting unit.
During the baseband conversion, an additional circuit noise, denoted as $v_i$, is present due to the phase offsets and 
non-linearities, which is modelled as an \acl{AWGN} with zero mean and variance $\sigma_c^2$.
Thus, the \ac{SINR} at device $k$ is given by
\begin{equation}
\Gamma_k^d = \frac{\delta_k \Abs{\Herm{\vtH}_k\vtV_k}^2 }{\delta_k \left(\sum_{j\neq k} \Abs{\Herm{\vtH}_k\vtV_j}^2 
+ \sigma^2 \right) + \sigma_c^2 }.
\end{equation}
In the \acl{DL}, the achievable rate at device $k$ is given by
\begin{equation}
R_k^d = B_c\log_2\left(1 + \Gamma_k^d\right).
\end{equation}

\subsection{Energy and Time Models}\label{sub:energy_model}
We define the energy consumption to compute the model at device $k$ as~\cite{Yang2021}
\begin{equation}\label{eq:energy_comp}
E_k^c = \kappa C_k A_k I_k f_k^2,
\end{equation}
where $\kappa$ is the effective switched capacitance, $C_k$ is the number of central processing unit cycles required 
for computing one sample data at device in Mcycles/bit, $A_k$ is the dataset size in bits at device $k$, $I_k$ is the 
number of local iterations performed in the training phase by the learning algorithm, and $f_k$ is the processor clock 
frequency at device $k$ in GHz.
Assuming that the \acl{UL} transmission occurs during $t_k^u$ seconds and the transmitted power is $p_k$, the energy 
to transmit the weights in the \acl{UL} is $E_k^t=t_k^u p_k$. 
Hence, the total energy spent at device $k$ during one communication round is $E_k^c + E_k^t$.

We consider a nonlinear energy harvesting model as in~\cite{Xu2017}, in which the harvested power $P_k^h$ is given by
\begin{align}
P_k^h = \alpha_1 \left((1-\delta_k)P_k^r \right)^2 + \alpha_2 \left((1-\delta_k)P_k^r \right) + \alpha_3,
\end{align}
where $\alpha_1,\alpha_2,\alpha_3\in\fdR$ are parameters of the model. 
We define the energy harvested at device $k$ as $E^h_k = t_k^d P_k^h$, where $t_k^d$ is the transmission time in 
seconds that the edge server spends to unicast the aggregated weight to device $k$.
Therefore, we constraint the total energy spent by the total harvested energy at device $k$, i.e., $E_k^t + E_k^c \leq 
E^h_k$.

In the \acl{UL}, we assume that device $k$ has $D_k$ bits to transmit, corresponding to the dimension of the 
weight updates multiplied by the number of bits used to represent each dimension.
Thus, each device $k$ has a constraint $t_k^u R_k^u \geq D_k$ in the \acl{UL}.
Similarly, the \acl{DL} rate must ensure that the aggregated weight is unicasted to each device.
Hence, each device $k$ has a constraint $t_k^d R_k^d\geq D_k$ in the \acl{DL}.

In addition, we need to take into account the time necessary for device $k$ to solve the local problem.
Let us denote the time to solve the local training problem device $k$ as~\cite{Yang2021}
\begin{align}\label{eq:time_compute_k}
t_k^c = C_k A_k I_k/f_k,
\end{align}
where $C_k$, $A_k$, $I_k$, and $f_k$ are the same parameters as in Eq.~\eqref{eq:energy_comp}.
Therefore, we denote the time for one communication round, including the \acl{UL} transmission with local computation, 
and \acl{DL} transmission with harvesting time, as $t^r = \max_k\; (t_k^u + t_k^c) + \max_k\;(t_k^d)$.

\subsection{Learning Model}\label{sub:learning_model}

The \ac{FL} methods~\cite{McMahan2017,Li2020b} are designed to support multiple devices collecting data and 
a server (edge server in our case) coordinating the global learning objective for the network. 
The objective of the network is to solve the problem below:
\begin{align}\label{eq:global_FL_probl}
	\underset{\vtW}{\text{minimize}}
	& f(\vtW) = \sum\nolimits_{k=1}^{K} \frac{N_k}{N} F_k(\vtW), 
\end{align}
where $N = \sum_{k=1}^{K} N_k$ and $\vtW$ is the global weight of the network.
The devices measure the local empirical risk over different data distributions $\calD_k$, i.e. 
$F_k(\vtW)\triangleq \fdE_{\mtX_k\sim \calD_k} [f_k(\vtW;\mtX_k)]$ in which the expectation is over the $N_k$ samples 
held by device $k$.
The function $f_k(\vtW;\mtX_k)$ is the loss function at device $k$ using its local dataset $\mtX_k$ and the global 
weight $\vtW$.
To reduce the communication of the weights for each device, \ac{FL} methods consider that each device 
$k$ minimizes the loss function $F_k(\vtW_k)\triangleq \fdE_{\mtX_k\sim \calD_k} [f_k(\vtW_k;\mtX_k)]$ by using its own 
data $\mtX_k$ and with its own weights $\vtW_k$.
Then, the devices send the resulting $\vtW_k$ to the edge server, which aggregates the weights from each device and 
sends back the global weight $\vtW = \sum_{k=1}^{K} (N_k/N) \vtW_k$.
The devices use the global weight as the initial point to minimize $F_k(\vtW_k)$, and this process continues until 
problem~\eqref{eq:global_FL_probl} converges to the optimal solution.

The \ac{FL} approach considered is \mbox{FedProx}~\cite{Li2020b}, which guarantees convergence for scenarios with 
heterogeneous devices and non-convex learning objectives.
Due to these advantages, \mbox{FedProx} suits particularly well \ac{IoT}.
To further reduce communication exchange among the edge server and the devices, the local objective function is a 
surrogate for the global objective function.
The surrogate function $h_k(\vtW_k;\vtW^t)$ includes a proximal term to the local subproblem with the global update 
$\vtW^t$ at communication round $t$.
Thus, each device $k$ solves the following local training problem
\begin{align}\label{eq:local_FL_probl}
	\underset{\vtW_k}{\text{minimize}}
	&\; h_k(\vtW_k;\vtW^t) = F_k(\vtW_k) + \frac{\mu}{2} \Norm{\vtW_k - \vtW^t}^2,
\end{align}
where $\mu>0$ is a local configurable parameter.
In addition to computing the surrogate function, each device may solve problem~\eqref{eq:local_FL_probl} inexactly so 
that a flexible performance depending on the local computation capacity and communication may be included in the 
solution. 
Specifically, we introduce the $\gamma_k^t$-inexactness for device $k$ at round $t$ below.

\vspace*{-0.2cm}
\begin{defi}[$\gamma_k^t$-inexact solution]\label{def:gamma_inexact}
For a function $h_k(\vtW_k;\vtW^t) = F_k(\vtW_k) + (\mu/2) \Norm{\vtW_k - 
\vtW^t}^2$, and $\gamma_k^t\in [0,1]$, we define $\vtW_k^{(n)}$ as the $\gamma_k^t$-inexact solution of problem
$\underset{\vtW_k}{\text{minimize}}\; h_k(\vtW_k;\vtW^t)$ at iteration $n$ if 
$$\Norm{\nabla h_k(\vtW_k^{(n)};\vtW^t)} \leq  \gamma_k^t\Norm{\nabla h_k (\vtW^t;\vtW^t) },$$ where 
$\nabla h_k(\vtW_k;\vtW^t) = \nabla F_k(\vtW_k) + \mu (\vtW_k -\vtW^t).$
\end{defi}
Hence, problem~\eqref{eq:local_FL_probl} is solved using \ac{SGD} and iterates until a $\gamma_k^t$-inexact solution is 
obtained. 
The iterations are local at device $k$ and are termed \emph{epochs}.
Different than the traditional federated averaging method~\cite{McMahan2017}, the number of epochs in \mbox{FedProx} is 
controlled by $\gamma_k^t$ and may be different for each device.

From~\cite[Theorem 4 and Corollary 9]{Li2020b}, the convergence of FedProx is as follows.

\vspace*{-0.2cm}
\begin{thm}[Convergence: Variable $\gamma$'s]\label{thm:conv_vargamma}	
Assume the functions $F_k$ are non-convex, $L$-Lipschitz smooth, and there exists $L_- > 0$, such that $\nabla^2 
F_k \succeq -L_- \mtI$, with $\bar{\mu} \triangleq \mu-L_->0$. 
Consider that $B$ is a global measure of dissimilarity between the gradients of the devices~\cite[Definition 
3]{Li2020b}, and suppose that $\vtW^t$ is not a stationary solution and the local functions $F_k$ are $B$-dissimilar.
If $\mu$, $K$, and $\gamma_k^t$ are chosen such that
\begin{align}\label{eq:rho_conv_gamma}
\rho &= \left(\frac{(1-\gamma^t B)}{\mu} - (1+\gamma^t)(a_1 + a_2(1+\gamma^t)) \right) > 0,
\end{align}
then at communication round $t$ of \mbox{FedProx}, we have the following expected decrease $\rho$ in the global 
objective 
$$\fdE_{S_t} [f(\vtW^{t+1})] \leq f(\vtW^{t}) - \rho \Norm{\nabla f(\vtW^t)}^2,$$ 
where $S_t$ is the set of $K$ devices selected at iteration $t$, $\gamma^t=\max_{k \in S_t}~{\gamma_k^t}$, and the 
constants $a_1,\; a_2$ are
\begin{align*}
a_1 = \frac{B\mu\sqrt{2} + L B\sqrt{K} }{\mu \bar{\mu}\sqrt{K}},\quad 
a_2 = \frac{LB^2}{2\bar{\mu}^2 K} (K+4\sqrt{2}K+4).\qed
\end{align*}
\end{thm}
In the following, we drop the index $t$ in $\gamma_k^t,\;\gamma^t$ and refer to them as $\gamma_k,\;\gamma$ to simplify 
notation, and consider that $S_t$ corresponds to all devices $K$. 
Moreover, the conditions in~\cite[Remark~5]{Li2020b} that ensure $\rho >0$, i.e., that sufficient decrease is 
attainable after each communication round, are the following
\begin{align}\label{eq:cond_gamma_rho_b}
\gamma B<1,\quad B<\sqrt{K}.
\end{align}
Then, the authors in~\cite[Theorem~6]{Li2020b} characterize the rate of convergence to the set of approximate 
stationary solutions $\calS_s = \{\vtW\;|\; \fdE[\Norm{\nabla f(\vtW) }^2] \leq \epsilon\}$ of 
problem~\eqref{eq:global_FL_probl}. 
Finally, the total number of communication rounds of \mbox{FedProx} is $T = O(\frac{\Delta}{\rho 
\epsilon})$, where $\Delta=f(\vtW^1)-f^*$.

In order to define the optimization problem, we first need to derive the number of local iterations (epochs) $I_k$, 
which depends on the solver used to obtain the optimal $\vtW_k^\star$ that provides the $\gamma_k$-inexact solution. 
In the following theorem, we establish the expression for $I_k$ in terms of $\gamma_k$.

\vspace*{-0.2cm}
\begin{thm}\label{thm:local_iter_number}
Consider that the local problem at device $k$ is solved via gradient descent with step size $\alpha < 2/(L+\mu)$.
Consider that the initial iteration for device $k$ is given by $\vtW_k^0 = \vtW^t$, and that $\beta = 
2/(\alpha\bar{\mu}\left(2-\alpha (L + \mu)\right))$.
Then, the number of local iterations $I_k$ is lower-bounded by
\begin{align}\label{eq:bound_loc_iter}
I_k \geq 2\beta\log\left(\frac{L + \mu }{\gamma_k^t \bar{\mu}}\right).
\end{align}
\end{thm}

\vspace*{-0.4cm}
\begin{proof}
See Appendix~\ref{app:prof_thm_iter_number}.
\end{proof}

\vspace*{-0.2cm}
Based on the theoretical guarantees for the convergence of FedProx, we can now properly formulate the problem in the 
next section.

\section{Minimization of Communication Rounds and Round Time}\label{sec:probl_form}
Our objective is to minimize the number of communication rounds between the edge server and the devices, while 
minimizing the time for each communication round and ensuring that sufficient energy is harvested by the devices.
We consider the total number of communication rounds between the edge server and the devices from 
Section~\ref{sub:learning_model}.
Given that the term $\Delta$ cannot be estimated beforehand, we disregard $\Delta$ and assume that the total number of 
communication rounds is $I_g(\rho,\epsilon) \!\!=\!\! 1/(\rho \epsilon)$.
Since we aim to minimize the number of communication rounds $I_g(\rho,\epsilon)$, we can maximize $\rho$, i.e., the 
expected decrease in the global objective function per communication round, due to its monotonic decreasing relation 
with $I_g(\rho,\epsilon)$.

We formulate the trade-off between the number of communication rounds and the communication round time using 
\ac{SWIPT} to harvest while learning as
\begin{subequations}\label{eq:min_comm_rounds}
\begin{align}
\underset{\substack{ \{t_k^u\},\{t_k^d\},\{\gamma_k\},\gamma } }{\text{minimize}}\quad
& t^r -\rho \label{eq:obj_max_rho}\\
\text{subject to} \quad & t_k^u R_k^u \geq D_k, \forall k, \label{eq:min_data_k}\\
\quad & t_k^d R_k^d\geq D_k, \forall k, \label{eq:min_data_server_k}\\
\quad& E^h_k \geq \zeta\left(E_k^t + E_k^c\right), \forall k, \label{eq:energy_const}\\
\quad& \gamma_k \leq \frac{1-\xi}{B}, \forall k, \label{eq:gamma_k_const}\\
\quad& \gamma \geq \gamma_k, \forall k, \label{eq:gamma_const}\\
\quad& t_k^u, t_k^d, \gamma_k \geq 0, \forall k, \label{eq:nonzero_var}
\end{align}
\end{subequations}
where the objective is to minimize the number of communication rounds and the communication round time.
The constraints~\eqref{eq:min_data_k}-\eqref{eq:min_data_server_k} state that all the data should be transmitted to 
the edge server and then received by device $k$; constraint~\eqref{eq:energy_const} states that at least a $\zeta\in 
(0,1]$ part of the total energy used by device $k$ should be harvested; constraint~\eqref{eq:gamma_k_const} states the 
maximum value for $\gamma_k$ is constrained according to Eq.~\eqref{eq:cond_gamma_rho_b}, in which $\xi\in [0,1]$ 
is a predefined constant to turn the strict inequality into a non-strict inequality; and 
constraint~\eqref{eq:nonzero_var} states that all variables should be nonzero.

\subsection{Solution Approach}\label{sub:sol_approach}
For ease of notation, we rewrite the terms that include $\gamma_k$ in the objective function and 
constraints of problem~\eqref{eq:min_comm_rounds}. 
First, we can rewrite the energy to compute the model at device $k$, $E_k^c$, in Eq.~\eqref{eq:energy_comp} as
\begin{align}\label{eq:energy_comp_mod}
E_k^c = a_{3,k}\left(\log(a_4)-\log(\gamma_k)\right),
\end{align}
where the parameters $a_{3,k},\; a_4$ denote the following
\begin{align*}
a_{3,k} = 2\kappa\beta C_k A_k f_k^2,\; \forall k, \quad a_4 = (L+\mu)/\bar{\mu}.
\end{align*}
Similarly, we denote the time to compute the model, $t_k^c$, as
\begin{align*}
t_k^c = a_{5,k}\left(\log(a_4)-\log(\gamma_k)\right),\quad a_{5,k} = 2\beta C_k A_k/f_k.
\end{align*}

With these reformulations, we can rewrite problem~\eqref{eq:min_comm_rounds} in an equivalent form as
\begin{subequations}\label{eq:learning_problem}
\begin{align}
\underset{\substack{ \gamma,\{\gamma_k\},\\ \{t_k^u\}, \{t_k^d\} } }{\text{minimize}}\quad
& \left(\max_k\left(a_{5,k}\left(\log(a_4/\gamma_k)\right) + t_k^u\right) + \max_k(t_k^d) \right) +\nonumber\\
&  (1+\gamma)(a_1 + a_2(1+\gamma)) - (1-\gamma B)/\mu\label{eq:learn_obj}\\
\text{subject to} \quad & \gamma \geq \gamma_k, \forall k, \label{eq:max_gamma_k}\\
\quad & t_k^d P_k^h \geq \zeta\left(t_k^u p_k + a_{3,k}\log\left(\frac{a_4}{\gamma_k}\right)\right), \forall k, 
\label{eq:energy_gamma}\\
\quad & \gamma_k \leq (1-\xi)/B, \forall k, \label{eq:gamma_k_linear}\\
\quad & t_k^u \geq D_k/R_k^u, \forall k \label{eq:time_k_UL_min} \\
\quad & t_k^d \geq D_k/R_k^d, \forall k \label{eq:time_DL_min} \\
\quad & \gamma_k, t_k^d, t_k^u \geq 0, \forall k.
\end{align}
\end{subequations}
Note that problem~\eqref{eq:learning_problem} is jointly convex on all variables because the objective function is 
concave, while the constraints are convex.
Hence, it can be solved by using traditional convex solvers, such as interior-point methods present on 
CVXPY~\cite{CVXPY2016}.
The computational complexity is $\calO(K^4)$~\cite{GahinetLMI1995}, which is due to the use of interior-point methods 
to solve problem~\eqref{eq:learning_problem}.
\section{Numerical Results}\label{sec:num_results}
We consider a system comprised by a single pico-cell~\cite{3gpp.36.828}.
The total number of antennas at the edge server is $M=16$ and the total number of devices is $K=10$. 
We consider an image classification learning task using the MNIST dataset and a multi-class logistic regression as loss 
function.
The samples per device $N_k$ are randomly assigned to the devices in a non-IID manner, which ensures that a device will 
have tens or thousands of samples, and consider $32$ bits encoding, i.e., $A_k = 32N_k$.
We consider \ac{MRC} at the receiver and two transmitter beamforming options namely, \ac{MRT} and \ac{ZF}.
The test accuracy is the performance indicator used for learning tasks because it represents the percentage of correct 
predictions.
The parameters necessary to obtain the numerical results are listed in \tabref{tab:sim_param}.

\begin{table}
	\caption{Simulation parameters}\label{tab:sim_param}
	\scriptsize
	\begin{tabularx}{\columnwidth}{l|l}
		\hline 
		\textbf{Parameter}                         & \textbf{Value} \\
		\hline 
		Cell radius                                & \unit{40}{m} \\
		Number of devices $K$ 					   & $10$  \\
		Number of antennas at \ac{BS} $M$          & $16$\\
		Monte Carlo iterations                     & $100$ \\ 
		Carrier frequency / Channel bandwidth $[B_c]$& \unit{2.45}{GHz}/\unit{180}{kHz} \\
		\ac{LOS}/\ac{NLOS} path-loss model         & Set according to~\cite[Table 6.2-1]{3gpp.36.828}\\
		Thermal and circuit noise power $[\sigma^2,
		\sigma^2_c]$         					   & $[-95, -60]$~\unit{}{dBm} \\ 
		\ac{BS}/Device Tx power $[\Norm{\vtV_k}^2, 
		p_{k}]$ 								   & $[30,-20]$~\unit{}{dBm}\\ 
		Learning parameters $[\alpha, L, \mu, 
		\bar{\mu}, \xi, B]$						   & $\scale[0.8]{[10^{-3}, 100, 10, 10, 0.05, (1-\xi)\sqrt{K}]}$\\
		Energy harvesting params. $[\alpha_1, 
		\alpha_2, \alpha_3, \delta_k]$             & $\scale[0.8]{[-0.1160, 0.6574, -6.549\!\times\! 10^{-7}, 0.1]}$\\
		Energy computing params. $[\kappa, f_k,C_k, 
		D_k]$                                     & $[10^{-28}, 0.1,\stU(10,30), 251200]$\\
		\hline
	\end{tabularx}
\end{table}

\begin{figure}
	\centering
	\includegraphics[width=0.8\linewidth,trim=0mm 0mm 0mm 0mm,clip]{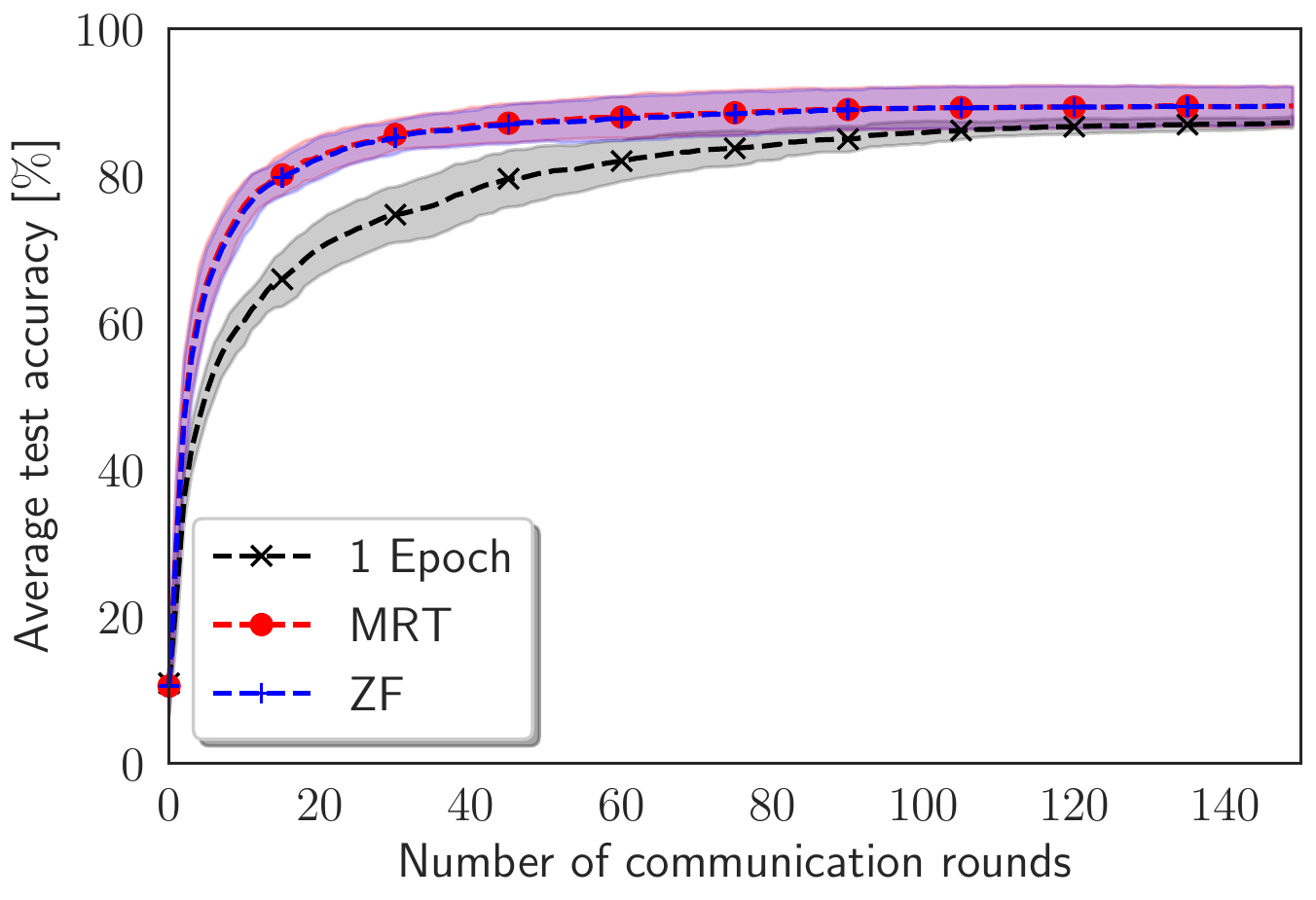}
	\caption{The average test accuracy versus the number of communication rounds for different $\gamma_k$ optimization, 
		including no optimization (1 epoch), \ac{MRT}, and \ac{ZF} optimization. \ac{MRT} and \ac{ZF} greatly 
		outperform the no 
		optimization baseline, showing the importance of an optimization of the local number of iterations based on the 
		solution of problem~\eqref{eq:learning_problem}.}
	\label{fig:test_accuracy}
\end{figure}
\figref{fig:test_accuracy} shows the average test accuracy versus the number of communication rounds over 10 Monte 
Carlo iterations using a baseline \mbox{FedProx} with only one epoch, i.e., without $\gamma_k$ optimization and using 
only one \ac{SGD} iteration.
We compare it with the solution of problem~\eqref{eq:learning_problem} by using \ac{MRT} and \ac{ZF} beamforming.
We assume that $\zeta=1.0$, which means that in each round the devices must harvest all the total energy spent with the 
training and transmission.
Note that the shaded region for the curves represent the standard deviation at each communication round.
Without $\gamma_k$ optimization, the average test accuracy increases much slower than the \ac{MRT} and \ac{ZF} 
solutions using $\gamma_k$ optimization.
For instance, at 20 communication rounds the average test accuracy of \ac{MRT} and \ac{ZF} are approximately 
\unit{82}{\%}, while with one epoch it is only \unit{69}{\%}.
Moreover, there is almost no difference in the average and standard deviation of the test accuracy for \ac{MRT} and 
\ac{ZF}, which is due to a similar number of local iterations from the solution of problem~\eqref{eq:learning_problem}.
The energy received by \ac{ZF} is lower than \ac{MRT} because \ac{ZF} nulls the interference from 
neighbouring users, and thus causes inequality~\eqref{eq:energy_gamma} harder to be attainable.
To fulfil inequality~\eqref{eq:energy_gamma}, the \acl{DL} time is increased and the value of $\gamma_k$ is decreased, 
but kept at similar values for both \ac{ZF} and \ac{MRT}.
Therefore, there are substantial benefits in terms of communication rounds, when optimizing the learning parameters 
$(\gamma_k)$ based on the solution of problem~\eqref{eq:learning_problem}.


\begin{figure}
	\centering
	\includegraphics[width=0.8\linewidth,trim=0mm 0mm 0mm 0mm,clip]{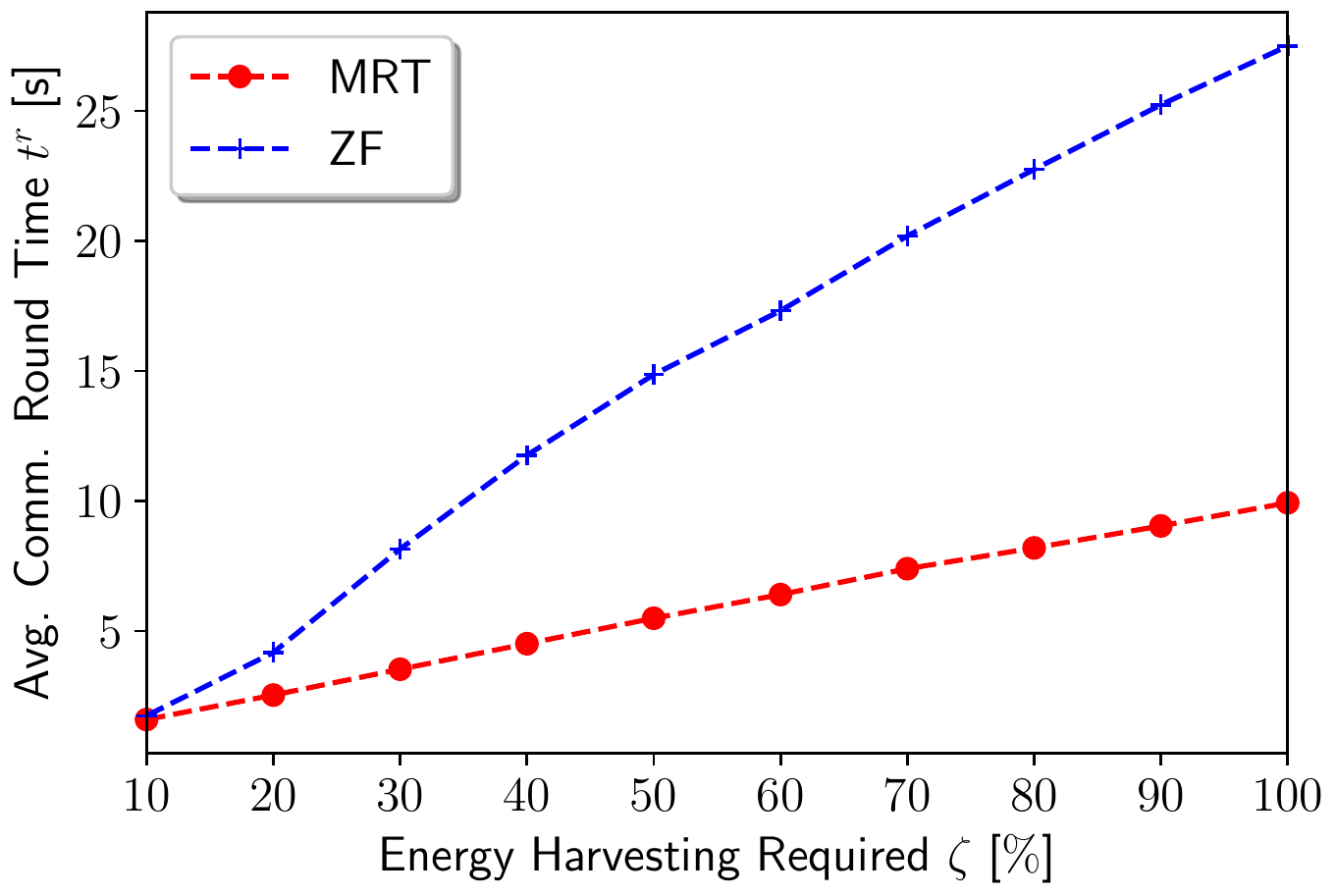}
	\caption{Average time for one communication round $(t^r)$ versus minimum energy harvesting required $(\zeta)$ with 
		\ac{MRT} and \ac{ZF} beamforming. Note that \ac{MRT} has a much lower communication round time than \ac{ZF}, 
		which is 
		due to much higher time needed to harvest sufficient energy in the \acl{DL}.}
	\label{fig:comm_round}
\end{figure}

\figref{fig:comm_round} shows the average communication round time, $t^r$, across all Monte Carlo iterations versus the 
percentage of energy harvesting required $\zeta$, when using \ac{MRT} and \ac{ZF} beamforming.
As $\zeta$ increases from \unit{10}{\%} to \unit{100}{\%}, the communication round time increases much faster for 
\ac{ZF} than \ac{MRT}.
It indicates that this difference is due to $t_k^d$ and that inequality~\eqref{eq:energy_gamma}, instead of 
inequality~\eqref{eq:min_data_server_k}, is the bottleneck for both \ac{MRT} and \ac{ZF}.
Since \ac{ZF} nulls the interference from neighbouring users, the total received and harvested power is lower than the 
\ac{MRT} case, as expected from the energy harvesting literature~\cite{Timotheou2014}.
Moreover, it implies that the necessary time to harvest a $\zeta$-percentage of the total energy used by device $k$ is 
much higher with \ac{ZF} than \ac{MRT}.
Since \ac{MRT} and \ac{ZF} have practically the same test accuracy, the \ac{MRT} beamforming is a better solution to 
use in a practical scenario because it provides the best trade-off in terms of test accuracy and energy harvesting 
capabilities within a short communication round time.

\section{Conclusion}\label{sec:concl}
In this paper, we considered the trade-off between the number of communication rounds in a \ac{FL} model and the 
necessary communication round time to harvest energy in an \ac{IoT} scenario.
Specifically, our objective was to minimize the total number of communication rounds necessary to learn a model using 
\ac{FL} while minimizing the communication round time, which includes the time to compute and transmit the model, and 
the time to harvest a sufficient amount of the energy spent during training and transmission.
Considering fixed transmit beamforming, this resulted in an optimization problem with learning and time parameters that 
is jointly convex in all the parameters.
The numerical results showed that the solution to this problem using either \ac{MRT} or \ac{ZF} beamforming achieves a 
much higher test accuracy than a baseline solution that does not optimize the learning parameters based on 
communication aspects.
We showed that the \ac{MRT} beamforming solution provides much lower communication round time than \ac{ZF} 
beamforming, while achieving practically the same test accuracy. 
Therefore, we showed that \ac{MRT} beamforming is able to achieve a much better trade-off in terms of number of 
communication rounds and communication round time.
For future work, we intend to study the impact of using optimized beamforming and energy harvesting parameters 
together with the learning and time parameters in the aforementioned \mbox{trade-off}.

\appendices

\section{Proof of Theorem~\ref{thm:local_iter_number}}\label{app:prof_thm_iter_number}
In the following proof, we refer to $h(\vtW_k;\vtW^t)$ as $h(\vtW_k)$, and refer to $\NormTwo{\cdot}$ as $\Norm{\cdot}$ 
to ease the notation.
From~\cite{Li2020b}, recall that the function $h(\vtW_k)$ is $\bar{\mu}$-strongly convex, and that the function 
$F_k(\vtW_k)$ is $L$-smooth.
First, we prove in Lemma~\ref{lem:Lmu_smooth_h} that the $h(\vtW_k)$ is also $(L+\mu)$-smooth.
\begin{lem}\label{lem:Lmu_smooth_h}
Consider the function $$h_k(\vtW_k) = F_k(\vtW_k) + \frac{\mu}{2} \Norm{\vtW_k - \vtW^t}^2,$$ where 
$F_k(\vtW_k)$ is $L$-smooth.
Then, the function $h_k(\vtW_k)$ is $(L+\mu)$-smooth.
\end{lem}

\begin{proof}
Using the definition of $h_k(\vtW_k)$ for any $\vtW_k$ as $\vtW$ and $\vtZ$, we have that
\begin{align*}
\nabla^2 h_k(\vtW) &= \nabla^2 \left( F_k(\vtW) + \frac{\mu}{2} \Norm{\vtW - \vtW^t}^2\right),\\
&= \nabla^2 F_k(\vtW) + \mu \mtI ,\\
&\overset{(a)}{\leq} L\mtI  + \mu \mtI = (L+\mu) \mtI,
\end{align*}
where $(a)$ is due to  $h_k(\vtW_k)$ being $L$-smooth.
Hence, we have proved that the function $h_k(\vtW_k)$ is $(L+\mu)$-smooth.
\end{proof}

Let us assume that $n$ is the last iteration of the local solver using gradient descent, i.e., the inequality 
$\Norm{\nabla h_k(\vtW_k^{(n)})} \leq  \gamma_k^t\Norm{\nabla h_k (\vtW^t) }$ is valid with 
$0 \leq \gamma_k^t\leq 1$.
As a local solver, we consider the traditional gradient descent algorithm with step size $\alpha$ and updates as
\begin{equation}\label{eq:grad_des_upd}
\vtW_k^{(n)} = \vtW_k^{(n-1)} - \alpha \nabla h_k(\vtW_k^{(n)}).
\end{equation}
We denote as $\vtW_k^\star$ the optimal (exact) solution of the optimization problem of minimizing $h_k(\vtW_k)$ over 
$\mathbf{w}_k$.

Since $h_k(\vtW_k)$ is $(L+\mu)$-smooth, we have the following inequality between iterations $n$ and 
$n-1$~\cite[Proposition (6.1.2)]{Bertsekas2015}:

{\small
\begin{align*}
h_k(\vtW_k^{(n)}) &\leq h_k(\vtW_k^{(n-1)}) + \nabla \Transp{h_k(\vtW_k^{(n-1)})} \left(\vtW_k^{(n)} - 
\vtW_k^{(n-1)} 
\right) + \\
&\hspace*{0.4cm} \frac{L+\mu}{2}\Norm{\vtW_k^{(n)} - \vtW_k^{(n-1)}}^{2},\\
&\overset{(c)}{\leq} h_k(\vtW_k^{(n-1)}) - \alpha \Norm{\nabla h_k(\vtW_k^{(n-1)})}^2 +\\ 
&\hspace*{0.4cm} \frac{\alpha^2 (L+\mu)}{2} \Norm{\nabla h_k(\vtW_k^{(n-1)})}^2,\\
&= h_k(\vtW_k^{(n-1)}) \!-\!\alpha\!\left(\!\frac{2- \alpha(L+\mu)}{2}\!\right)\!\! \Norm{\nabla 
	h_k(\vtW_k^{(n-1)})}^2,
\end{align*}
}
where $(c)$ is due to Eq.~\eqref{eq:grad_des_upd}. 

Hence, we have the following inequality between iterations $n$ and 
$n-1$:
{\small
\begin{equation}\label{eq:ineq_iter_n_n-1}
\hspace*{-0.2cm}h_k(\vtW_k^{(n)}) \!\leq  h_k(\vtW_k^{(n-1)}) \!-\! \alpha\left(\!\frac{2- 
\alpha(L+\mu)}{2}\!\right)\! 
\Norm{\nabla h_k(\vtW_k^{(n-1)})}^2\hspace*{-0.2cm}.
\end{equation}}

Recall that the function $h_k(\vtW_k)$ is $\bar{\mu}$-strongly convex. Using the fact that $\nabla h_k(\vtW_k^\star) = 
\vtZero$, the following inequality holds~\cite[Eq. (6.20)]{Bertsekas2015}:
\begin{equation*}
\Norm{\nabla h(\vtW_k^{(n-1)}) }^2 \geq \bar{\mu} \left(h_k(\vtW_k^{(n-1)}) - h_k(\vtW_k^\star) \right).
\end{equation*}
Using the expression above in Eq.~\eqref{eq:ineq_iter_n_n-1}, we have that:
{\small
\begin{align}
h_k(\vtW_k^{(n)}) &\leq  h_k(\vtW_k^{(n-1)}) \!-\! \alpha\bar{\mu}\left(\!\frac{2- \alpha(L+\mu)}{2}\!\right)\! 
\times\nonumber\\
&\hspace*{0.4cm} \left(h_k(\vtW_k^{(n-1)}) - h_k(\vtW_k^\star) \right),\nonumber\\
h_k(\vtW_k^{(n)}) - h_k(\vtW_k^\star) &\overset{(d)}\leq \left(h_k(\vtW_k^{(n-1)}) - h_k(\vtW_k^\star) \right) 
-\nonumber\\
&\hspace*{0.4cm}\! \alpha\bar{\mu}\!\left(\!\frac{2- \alpha(L+\mu)}{2}\!\right)\!\! \left(\!h_k(\vtW_k^{(n-1)}) 
\!-\! 
h_k(\vtW_k^\star) \!\right),\nonumber\\
h_k(\vtW_k^{(n)}) - h_k(\vtW_k^\star) &\overset{(e)}\leq \left[1 - \!\left(\!\frac{\left(2- 
	\alpha(L+\mu)\right)\alpha\bar{\mu}}{2}\!\right) \right] \times\nonumber\\
&\hspace*{0.4cm}\left(\!h_k(\vtW_k^{(n-1)}) \!-\! h_k(\vtW_k^\star) \!\right),\nonumber\\
h_k(\vtW_k^{(n)}) - h_k(\vtW_k^\star) &\overset{(f)}\leq \left[1 - \!\left(\!\frac{\left(2- 
	\alpha(L+\mu)\right)\alpha\bar{\mu}}{2}\!\right) \right]^n \times\nonumber\\
&\hspace*{0.4cm}\left(\!h_k(\vtW_k^{(0)}) \!-\! h_k(\vtW_k^\star) \!\right),\nonumber\\
h_k(\vtW_k^{(n)}) - h_k(\vtW_k^\star) &\overset{(g)}\leq \exp\left[-n\alpha\bar{\mu}\frac{\left(2- 
	\alpha(L+\mu)\right)}{2}\!\right] \times\nonumber\\
&\hspace*{0.4cm}\left(\!h_k(\vtW_k^{(0)}) \!-\! h_k(\vtW_k^\star) \!\right),\label{eq:bound_k_0}
\end{align}}
where $(d)$ adds the term $-h_k(\vtW_k^\star)$ to both sides of the inequality; $(f)$ uses the inequality in $(e)$ in a 
recursive manner from $n-1$ to $0$; and $(g)$ uses the inequality $(1-x)^n\leq \exp(-nx)$.

With Eq.~\eqref{eq:bound_k_0}, we need to bound the left- and right-hand sides such that we obtain $\gamma_k^t$.
Let us first bound the left-hand side using the $(L+\mu)$-smoothness of $h_k(\vtW_k)$ at $\vtW_k^{(n)}$ and 
$\vtW_k^\star$~\cite[Proposition 6.1.9]{Bertsekas2015}:

{\small
\begin{align}
h_k(\vtW_k^{(n)}) &\geq h_k(\vtW_k^\star) + \Transp{\nabla h_k(\vtW_k^\star)} \left(\vtW_k^{(n)} -\vtW_k^\star 
\right) 
+\nonumber\\ 
&\hspace*{0.4cm} \frac{1}{2\left(L+\mu \right)} \Norm{\nabla h_k(\vtW_k^{(n)} ) - \nabla 
	h_k(\vtW_k^\star)}^2,\nonumber\\
h_k(\vtW_k^{(n)}) - h_k(\vtW_k^\star) &\overset{(h)}\geq \frac{1}{2\left(L+\mu \right)} \Norm{\nabla 
h_k(\vtW_k^{(n)} 
	)}^2,\label{eq:bound_left_k_0}
\end{align}}
where $(h)$ uses the fact that $\nabla h_k(\vtW_k^\star) = \vtZero$.

For the right-hand side, we need to use $(L+\mu)$-smoothness~\cite[Proposition (6.1.2)]{Bertsekas2015} at 
$\vtW_k^{(0)}$ and $\vtW_k^\star$:
\begin{align}
h_k(\vtW_k^{(0)}) &\leq h_k(\vtW_k^\star) + \nabla h_k(\vtW_k^\star) \left(\vtW_k^{(0)} -\vtW_k^\star \right) 
+\nonumber\\ 
&\hspace*{0.4cm} \frac{L+\mu}{2} \Norm{\vtW_k^{(0)} -\vtW_k^\star }^2,\nonumber\\
h_k(\vtW_k^{(0)}) - h_k(\vtW_k^\star) &\overset{(j)}\leq \frac{L+\mu}{2} \Norm{\vtW_k^{(0)} -\vtW_k^\star 
}^2,\label{eq:bound_right_smooth}
\end{align}
where $(j)$ uses the fact that $\nabla h_k(\vtW_k^\star) = \vtZero$.
Then, we need to use the $\bar{\mu}$-strong convexity of $h_k(\vtW)$~\cite[Eq. (6.20)]{Bertsekas2015} 
at $\vtW_k^{(0)}$ and $\vtW_k^\star$ as follows:
\begin{align}\label{eq:bound_right_strong}
\Norm{\vtW_k^{(0)} -\vtW_k^\star } \leq \frac{1}{\bar{\mu}} \Norm{\nabla h_k(\vtW_k^{(0)} )}.
\end{align}
Combining both Eqs.~\eqref{eq:bound_right_smooth}-\eqref{eq:bound_right_strong}, we get that:
\begin{align}
h_k(\vtW_k^{(0)}) - h_k(\vtW_k^\star) &\leq \frac{L+\mu}{2\bar{\mu}^2} \Norm{\nabla h_k(\vtW_k^{(0)} )}^2,\nonumber\\
h_k(\vtW_k^{(0)}) - h_k(\vtW_k^\star) &\overset{(k)}\leq \frac{L+\mu}{2\bar{\mu}^2} \Norm{\nabla F_k(\vtW^t 
	)}^2,\label{eq:bound_right_k_0}
\end{align}
where $(k)$ uses the assumption that $\vtW_k^{(0)}=\vtW^t$.

Using the bounds for left-hand side in Eq.~\eqref{eq:bound_left_k_0} and for the right-hand side in 
Eq.~\eqref{eq:bound_right_k_0} of Eq.~\eqref{eq:bound_k_0}, we have that:
\begin{align}
\frac{1}{2\left(L+\mu \right)} \Norm{\nabla h_k(\vtW_k^{(n)})}^2 & \leq 
\exp\left[-n\alpha\bar{\mu}\frac{\left(2- \alpha(L+\mu)\right)}{2}\!\right] \times\nonumber\\
&\hspace*{0.4cm} \frac{L+\mu}{2\bar{\mu}^2} \Norm{\nabla F_k(\vtW^t )}^2,\nonumber\\
\Norm{\nabla h_k(\vtW_k^{(n)})}^2 &\overset{(l)}\leq \left(\!\! \left(\frac{L+\mu}{\bar{\mu}}\right)
\exp\left(\frac{-n}{2\beta}\!\right)\!\! \right)^2\!\!\times\nonumber\\
&\hspace*{0.4cm} \Norm{\nabla F_k(\vtW^t )}^2\!\!,\label{eq:bound_sim_gamma_t}
\end{align}
where $(l)$ reorganizes the terms with the definition of $\beta \!=\! \tfrac{2}{\alpha\bar{\mu}\left(2-\alpha (L + 
	\mu)\right)}$, with $\alpha\! <\! \tfrac{2}{L +\mu}$.
To ensure a $\gamma_k^t$-inexact solution as in Definition~\ref{def:gamma_inexact}, and using 
Eq.~\eqref{eq:bound_sim_gamma_t}, we need to ensure that:
\begin{equation*}
\left(\frac{L+\mu}{\bar{\mu}}\right) \exp\left(\frac{-n}{2\beta}\!\right) \leq \gamma_k^t.
\end{equation*}
If we reorganize the inequality above, we have that:
\begin{align}
\exp\left(\frac{-n}{2\beta}\!\right) &\leq \frac{\gamma_k^t \bar{\mu}}{L+\mu},\nonumber\\
\frac{-n}{2\beta} &\overset{(m)}\leq \log\left(\frac{\gamma_k^t \bar{\mu}}{L+\mu} \right),\nonumber\\
n &\overset{(o)}\geq 2\beta \log\left(\frac{L+\mu}{\gamma_k^t \bar{\mu}} \right),
\end{align}
$(m)$ applies the logarithm function on both sides of the inequality; and $(o)$ reorganizes the term in the logarithmic 
function such that the sign on both sides of the inequality is reordered.
By redefining $n$ as $I_k$, we have proved the bound of the number of iterations proposed in 
Theorem~\ref{thm:local_iter_number}.

\bibliographystyle{IEEEtran}
\bibliography{FLbio}

\end{document}